\documentclass[11pt]{article}
\usepackage{amsfonts,amsmath,amsthm,
	amssymb,
	graphicx,mathrsfs}

\usepackage[top=2.2cm,bottom=2.2cm,
			left=2.5cm,right=2.5cm]{geometry}
			
\usepackage[colorlinks=true,urlcolor=blue, linkcolor =blue,citecolor=blue,
breaklinks=true]{hyperref}

\usepackage{algorithm,algorithmicx,breakcites}
\usepackage{booktabs}
\usepackage{bbm} 

\theoremstyle{plain}
\newtheorem{theorem}{Theorem}
\newtheorem{definition}{Definition}
\newtheorem{lemma}{Lemma}
\newtheorem{assume}{Assumption}
\newtheorem{remark}{Remark}
\newtheorem{corollary}[theorem]{Corollary}
\newtheorem{Proposition}[theorem]{Proposition}

\newcommand{\bx}{ {\boldsymbol x} }
\newcommand{\by}{ {\boldsymbol y} }

\allowdisplaybreaks

\begin{document}

\title{\bf A reduced-rank approach to predicting multiple binary responses through machine learning}

\author{The Tien Mai}

\date{
\begin{small}
Department of Mathematical Sciences, 
Norwegian University of Science and Technology,
Trondheim 7034, Norway.
\\
Email: the.t.mai@ntnu.no
\end{small}
}

\maketitle

\begin{abstract}
This paper investigates the problem of simultaneously predicting multiple binary responses by utilizing a shared set of covariates. Our approach incorporates machine learning techniques for binary classification, without making assumptions about the underlying observations. Instead, our focus lies on a group of predictors, aiming to identify the one that minimizes prediction error. Unlike previous studies that primarily address estimation error, we directly analyze the prediction error of our method using PAC-Bayesian bounds techniques. In this paper, we introduce a pseudo-Bayesian approach capable of handling incomplete response data. Our strategy is efficiently implemented using the Langevin Monte Carlo method. Through simulation studies and a practical application using real data, we demonstrate the effectiveness of our proposed method, producing comparable or sometimes superior results compared to the current state-of-the-art method.
\end{abstract}

\paragraph*{Keywords:} binary responses, low-rank predictors, PAC-Bayesian inequalities, Langevin Monte Carlo, missing data.

\section{Introduction}
The relationship between multiple response variables and a set of predictors has been a topic of ongoing research and interest in the literature, with numerous studies dedicated to understanding and exploring this connection. One area of particular interest in this field is the use of reduced rank regression, which involves using a low-rank constraint to linearly connect the response variables and the predictors, see e.g. \cite{izenman2008modern,cook2018introduction,reinsel2023multivariate,giraud2021introduction}. This approach has been widely studied and applied, with numerous works published on the topic, including those by \cite{anderson1951estimating,izenman1975reduced,bunea2011optimal,geweke1996bayesian}, and many others. From frequentist to Bayesian approaches, there have been a wide range of methods and techniques employed to analyze and model these relationships, \cite{corander2004bayesian,chakraborty2020bayesian,alquier2013bayesian,goh2017bayesian,yang2020fully,kleibergen2002priors,chen2013reduced,she2017robust}.

However, despite the extensive research in this area, most studies have focused on real-valued responses. In many applications, the entries of the response matrix are binary, that is, they are in the set  $ \{ -1 ,1\} $. For example, the treatment responses from multiple drugs can be recoded as binary or categorical when measured from each cell line, as seen in studies by \cite{hayes2006racial,greenlund2005using,mishra2022negative}. This highlights the need for further research on the use of binary or categorical response variables in reduced rank regression and other multivariate modeling techniques.

The problem of modeling multiple binary response variables has received some limited attention in recent years, with few studies proposing reduced-rank regression models as a solution. One notable example is the paper by Luo et al. \cite{luo2018leveraging}, which proposed a mixed-outcome reduced-rank regression model to handle response matrices that include both binary and count data, and also addressed the issue of missing data in the responses. Another recent study, carried out independently of \cite{luo2018leveraging}, is the paper by Park et al. \cite{park2022low}, which additionally considered row-wise sparse constraints in addition to the low-rank assumption, but only for fully observed binary response matrices. The main idea behind these studies is to assume a marginal logistic regression model to relate the binary response and the covariates, and then employ a penalized maximum likelihood method. 

However, the studies in \cite{luo2018leveraging} and \cite{park2022low} primarily focus on recovering the parameter matrix of interest, and provide estimation error rates for their estimators. While their results demonstrate that their estimators are consistent when estimating a low-rank matrix, they do not provide any guarantee on the prediction error or misclassification error. This highlights the need for further research in this area, particularly in terms of developing methods that not only accurately estimate the parameter matrix, but also provide guarantees on the performance of predictions and classification.

One of the key challenges in addressing the problem of multiple binary responses is the lack of robust and generalizable models that can accurately predict and classify the outcomes. In this paper, we aim to address this gap by taking a machine learning approach and adopting a classification-based method to deal with binary output. Unlike traditional methods that rely on parametric models, we will consider a set of prediction matrices and seek to find the one that yields the best prediction error. This approach is built on the principles of statistical learning theory \cite{vapnik}, where the zero-one loss is used as a measure of prediction error, and the risk of the classifier is controlled by a PAC (probably approximately correct) bound. However, the non-convex nature of the zero-one loss function makes it computationally intractable. An alternative is the use of a convex surrogate such as the hinge loss, which was introduced in \cite{zhang2004statistical}, has been shown to be effective in a variety of machine learning tasks and has the added benefit of being computationally efficient. By leveraging this method, we aim to provide a robust and generalizable model for predicting and classifying multiple binary responses.

In this work, we propose a novel approach to addressing the problem of multiple binary responses that combines elements of both Bayesian and machine learning methodologies. Specifically, we propose a pseudo-Bayesian approach that utilizes a notion of risk based on the hinge loss, rather than relying on a likelihood function. Our approach is based on the principles of PAC-Bayesian theory \cite{STW,McA,herbrich2002pac,catonibook,dalalyan2008aggregation,seldin2012pac,AlquierRidgway2015,seldin2010pac,germain2015risk}, which provides theoretical guarantees on the prediction and misclassification error for our method. It is worth mentioning that using loss functions in replacing the likelihood is becoming popular in the so-called generalized Bayesian inference in recent years as documented for example in \cite{matsubara2022robust,jewson2022general,yonekura2023adaptation,medina2022robustness,grunwald2017inconsistency,bissiri2013general,lyddon2019general,syring2019calibrating}.

The use of a hinge loss-based risk function allows us to overcome some of the limitations of traditional likelihood-based Bayesian models, particularly when dealing with binary response variables. Unlike traditional likelihood functions, which can be difficult to model and computationally intensive to compute, the hinge loss function is convex and can be easily optimized. To further improve the efficiency and practicality of our proposed approach, we also develop an efficient gradient-based sampling method based on Langevin Monte Carlo. This method allows us to approximate the computation of our proposed method, making it more computationally tractable and suitable for large-scale applications. Overall, our proposed approach offers a novel and promising approach to modeling and predicting multiple binary responses, providing both theoretical guarantees and practical computational methods for its implementation.

Similar to the approach proposed in \cite{luo2018leveraging}, our proposed method has the capability to handle incomplete response matrices. This is achieved by extending our approach to account for missing data in the response matrix, which allows for greater flexibility and applicability of our method. The extension of our method to handle missing data is relatively simple and does not introduce any additional complexity to the overall approach. This allows for our method to be applied to a wider range of datasets with incomplete or missing data. This capability is especially useful in real-world applications, where missing data is often present and traditional methods may struggle to handle such cases effectively. For example, in studies involving medical treatments, it is not uncommon for certain patients to drop out of the study, resulting in missing data in the response matrix. Our proposed method allows for the inclusion of such missing data, providing a more comprehensive and realistic analysis of the treatment outcomes. Additionally, in observational studies, missing data can be a common problem due to various factors such as non-response, measurement error or data collection issues. Our proposed method's ability to handle missing data would be beneficial in these scenarios as well.

The remainder of the paper is structured as follows. Section \ref{sc_blr} provides a detailed description of the problem statement and presents our proposed method, along with its theoretical results. An extension to handle incomplete response data is also discussed in this section. In Section \ref{sc_numstudy}, we describe the Langevin Monte Carlo method used to compute our proposed method and present numerical studies on both synthetic and real datasets to demonstrate its performance. Conclusions and discussions are given in Section \ref{sc_conclus}. All technical proofs are gathered in Appendix \ref{sc_proof}.

\section{Problem and method}
\label{sc_blr}
\subsection{Problem statement}
We formally consider the following multiple binary responses with a set of common covariates problem: for units $ i = 1,\ldots, n $ with covariate vectors $ x_i \in \mathbb{R}^p $, there exist $ q $ binary responses $ y^k_{i} \in \{-1,1\} $ for $ k = 1, \ldots, q $. From a classification perspective, it would be natural to use a linear predictor as a function from $ \mathbb{R}^p $ to $ \{-1,1\} $ in the following way: when $ \bx_i $ is revealed, $M^{(k)} \in \mathbb{R}^p $ predicts $ \by_{i}^k $ by ${\rm sign}(\bx_i M^{(k)})$. 

For the matrix notation, let's define the binary response matrix $ Y = [ y^1, \ldots, y^q ] \in \{-1,1\}^{n\times q} $ and the covariate matrix $ X = [x_1^\top, \ldots, x_n^\top]^\top \in \mathbb{R}^{n\times p}  $. With multiple responses, the predictors can be written in matrix-form as $ M = [ M^{(1)}, \ldots, M^{(q)} ] \in \mathbb{R}^{p\times q} $.

The ability of the predictor to predict a new entry of the matrix is then assessed by the risk
\begin{align*}
	R(M) 
	= 
\mathbb{E} \left[ \mathbbm{1}_{( Y_{11} \cdot (X M)_{11} < 0)}\right],
\end{align*}
and its empirical counterpart is:
\begin{align*}
	r(M) 
	= 
	\frac{1}{nq}\sum_{i=1}^n \sum_{j=1}^q 
	\mathbbm{1}_{(Y_{ij} (XM)_{ij}<0)} .
\end{align*}
From the standard approach in classification
theory~\cite{vapnik,devroye1997probabilistic}, the best possible classifier is the Bayes classifier, $ M^B $, such that
$$
R(M^B)= \inf_M R(M).
$$

The anticipated property of the Bayes matrix, $ M^B $, is that it exhibits a low-rank structure or can be effectively approximated by a low-rank matrix.

For the sake of simplicity, we put $ \overline{R} = R(M^B) $ and $\overline{r}=r(M^B)$. The goal is to find an estimator $ \hat{M} $ that yields the minimal excess risk $ R(\hat{M}) -\overline{R} $.

While the risk $R(M)$ has a clear interpretation, working with its empirical counterpart $r(M)$ is challenging as it is non-smooth and non-convex. A common approach to overcome this issue is to replace the empirical risk with a convex surrogate~\cite{zhang2004statistical}. In this paper, we primarily focus on the hinge loss, which results in the following hinge empirical risk:
\begin{align*}
	r^h(M) 
	= 
	\frac{1}{nq}\sum_{i=1}^n \sum_{j=1}^q ( 1 - Y_{ij} (XM)_{ij})_+ \, ,
\end{align*}
where $ (a)_+ := \max(a,0),\forall a \in \mathbb{R} $.

\subsection{Estimation procedure}
Building upon the work previously done in the field of PAC-Bayesian theory, we define the pseudo-posterior distribution as follows:
\begin{align*}
\widehat{\rho}_\lambda(M)
\propto
\exp[-\lambda r^h(M)] \pi(M)
\end{align*}
where $\lambda>0$ is a tuning parameter that will be discussed later and $\pi(M)$ is a prior distribution, given in \eqref{prior_scaled_Student}, that promotes (approximately) low-rankness on the parameter matrix $M$.

The term $ \widehat{\rho}_\lambda $, known as the Gibbs posterior, can be interpreted as the posterior distribution under a Bayesian framework, where $\pi$ represents the prior distribution for the parameter $M$. However, this Bayesian interpretation is not essential for understanding the approach, as it relies on the proportionality of $\exp[-\lambda r^h(M)]$ to a likelihood function. The motivation behind defining $ \widehat{\rho}_\lambda $ stems from the minimization problem in Lemma \ref{lemma_donvara} rather than Bayesian principles. It is not necessary to have a likelihood function or a complete model; only the empirical risk based on the hinge loss function is required.

Nonetheless, we still refer to $\pi$ as the prior and $ \widehat{\rho}_\lambda $ as the pseudo-posterior. The measure $ \widehat{\rho}_\lambda $ can be seen as an adjusted version of $\pi$. Comparing two parameters, $m_1$ and $m_2$, if $ r^h(m_1) < r^h(m_2) $, then $\exp[-\lambda r^h(m_1)] > \exp[-\lambda r^h(m_2)] $ for any $ \lambda >0 $. This implies that, relative to $\pi$, $ \widehat{\rho}_\lambda $ assigns more weight to $ m_1 $ than to $ m_2 $. The adjustment in the distribution thus favors the parameter value that results in a smaller in-sample hinge empirical risk. The tuning parameter $ \lambda $ controls the degree of adjustment. The choice of $ \lambda $ will be further explored in subsequent sections. This pseudo-Bayesian approach has been previously studied in various low-matrix estimation problems, such as \cite{alquier2013bayesian,mai2017pseudo,mai2015,cottet20181,mai2023low}.

In this work, we have opted to use a spectral scaled Student prior distribution, as follows, with a parameter $\tau>0$,
\begin{align}
\label{prior_scaled_Student}
\pi(M)
\propto
\det (\tau^2 \mathbf{I}_{p} + MM^\intercal )^{-(p+q+2)/2}.
\end{align}

\noindent This prior can induce low-rankness of matrices $ M $, as it can be verified that
$
\pi(M)
\propto
\prod_{j=1}^{p} (\tau^2 + s_j(M)^2 )^{- (p+q+2)/2 }
$, where $ s_j(M) $ denotes the $j^{th}$ largest singular value of $ M $. It means that this prior follows a scaled Student distribution evaluated at $ s_j(M) $ which induces approximately sparsity on the $s_j(M)$ \cite{dalalyan2012sparse}. Thus, under this prior distribution, most of the $s_j(M)$ are close to $0$ and that $M$ is approximately low-rank. This prior has been used before in image denoising \cite{dalalyan2020exponential}, bilinear regression \cite{mai2023bilinear} and \cite{yang2018fast} for matrix completion. Even though this prior distribution is not conjugate in our problem, it is advantageous to utilize the Langevin Monte Carlo, a sampling method that relies on gradients for implementation purposes.

\subsection{Theoretical results}
\label{sec:theoretical_bound}

\begin{assume} 
	\label{assume_trueBaye}
	We assume that $ r^h(M^B) \leq 2r(M^B) $.
\end{assume}
\noindent
The above assumption can be relaxed to that there exist a positive constant $ C'>0 $ such that $ r^h(M^B) \leq (1+C')r(M^B) $. All the theoretical results in this work are subjected to this assumption.

In this work, we make use of the Mammen and Tsybakov's margin assumption in~\cite{mammen1999smooth}.

\begin{assume}[Margin assumption] 
		\label{assume_margin}
	We assume that there is a constant $C \geq 1 $ such that: 
	\begin{equation*}
	\mathbb{E}\left[\left(\mathbbm{1}_{Y (XM)\leq 0} - 
	\mathbbm{1}_{Y (XM^B)\leq 0} \right)^2\right] \leq 
	C[R(M)-\overline{R}].
	\end{equation*}
\end{assume}
\noindent As an example, in the noiseless case where $Y = {\rm sign} (XM^B) $ almost surely,
we have that
\begin{equation*}
\mathbb{E}\left[\left(\mathbbm{1}_{Y (XM)\leq 0} 
- 
\mathbbm{1}_{Y (XM^B)\leq 0} \right)^2\right]
= 
\mathbb{E}\left[\mathbbm{1}_{Y (XM) \leq 0}^2\right]
= 
\mathbb{E}\left[\mathbbm{1}_{Y (XM) \leq 0}\right] = R(M) = R(M)-\overline{R}.
\end{equation*}
Thus, the margin assumption is satisfied with $C=1 $.

We now present a theoretical bound on the expected risk for a random estimator of $ M $ generated from the pseudo-posterior $ \widehat{\rho}_\lambda(M) $.

\begin{theorem}
	\label{th:theo}
	Assume that Assumption \ref{assume_margin} is satisfied and put $ r^* = {\rm rank} (M^B ) $ 	and with $ \tau^2 =  (p+q)/(2q^2pn \|X\|_{_F}^2 ) $. Then, 
	for any $\epsilon \in (0,1) $ and for $\lambda = 
2 nq/(3C + 2) $, $ \varsigma \in(0,1) $,
	with probability at least $1-2\epsilon$,
	\begin{align*}
	\int R d\widehat{\rho}_\lambda 
	\leq 
	2.5\overline{R} + \Xi_{C,\varsigma} 
	\frac{ r^* (q+p+2)	
		\log \left( 
		1+ \frac{q\| X\|_{_F} \| M^B \|_{_F} \sqrt{np} }{ \sqrt{(p+q) r^* }} \right) 
	+
	\log(1/\epsilon) }{nq},
	\end{align*}
	where $ \Xi_{C,\varsigma}$ is a known constant that depends only on $\varsigma,C $.
\end{theorem}

The proof of the above theorem is given in Appendix \ref{sc_proof}. The technical argument used in the proof is known as ``PAC-Bayesian bounds", introduced in \cite{STW,McA} as a way to provide empirical bounds on the prediction risk of Bayesian-type estimators. However, it is well known that the PAC-Bayesian approach also comes with a set of powerful technical tools to establish non-asymptotic bounds as documented in \cite{catoni2003pac,catoni2004statistical,catonibook} that have been explored in this paper. For an in-depth exploration of PAC-Bayes bounds, including recent surveys and advancements, readers are encouraged to refer to the following references \cite{guedj2019primer} and \cite{alquier2021user}.

\begin{remark}
It is important to mention that the result of the above theorem has an adaptive characteristic in the sense that the estimator does not depend on the rank  $ r^* = {\rm rank}(M^B) $. When the rank $ r^* $ is small, the prediction error will be similar to the Bayes error, $ \overline{R} $, even with a small sample size. This type of result is commonly referred to as an `oracle inequality' as it suggests that our estimator performs as well as if we had access to the rank of $ M^B $ through an oracle. Additionally, it is noteworthy that $ r^* \neq 0 $ is not a necessary condition in the above formula. If $r^* =0 $, we interpret $0\log(1+0/0) $ as $0 $. 
\end{remark}

\begin{corollary}
	\label{co:theo}
	In the  case that $Y={\rm sign}(XM^B)$ a.s.,
	for any $ \epsilon \in (0,1) $ and for $\lambda = 2 nq/5 $,
	with probability at least $1-2\epsilon$,
	\begin{align}
	\int R d\widehat{\rho}_\lambda 
	\leq 
	\Xi_{1,\varsigma}'
		\frac{ r^* (q+p+2)	\log \left( 
		1+ \frac{q\| X\|_{_F} \| M^B \|_{_F} \sqrt{np} }{ \sqrt{(p+q) r^* }} \right) 
		+
		\log(1/\epsilon) }{nq}
	\end{align}
	where $ \Xi_{C,\varsigma}' = \Xi_{1,\varsigma} $.
\end{corollary}

\begin{remark}
The Theorem \ref{th:theo} and Corollary \ref{co:theo} presented in this study offer novel perspectives on the prediction error, which  complement the previously established theoretical results on estimation errors in \cite{luo2018leveraging} and \cite{park2022low}. These theoretical inequalities allow for the comparison of the out-of-sample error of our predictor to the optimal one, and demonstrate that the prediction error rate is $ r^* \max(q,p)/nq $, with logarithmic terms included. This is a noteworthy contribution to the field as it provides a comprehensive understanding of the relationship between the rank of $M^B$ and the prediction error.
\end{remark}

It is worth mentioning that the utilization of Assumption \ref{assume_margin} plays a crucial role in achieving a 'fast' prediction rate, as demonstrated in Theorem \ref{th:theo}. The initial introduction of this assumption was made in \cite{mammen1999smooth} for classification purposes, and it has since been adopted for ranking tasks in subsequent works such as \cite{clemenccon2008ranking,robbiano2013upper,ridgway2014pac}. In the forthcoming proposition, we present a slower rate result without relying on the usage of Assumption \ref{assume_margin}. The proof is also given in Appendix \ref{sc_proof}.

\begin{Proposition}
	\label{thm_propotion}
	Put $ r^* = {\rm rank} (M^B ) $ 	and with $ \tau^2 =  (p+q)/(2q^2pn \|X\|_{_F}^2 ) $. Then, 
	for any $\epsilon \in (0,1) $ and for $\lambda = 
 2\sqrt{nq/(p+q+2)} $, $ \varsigma \in(0,1) $,
	with probability at least $1-2\epsilon$,
	\begin{align*}
	\int R d\widehat{\rho}_\lambda 
	\leq 
	2\overline{R} + \Psi_{\varsigma} 
	\frac{ r^* \sqrt{(q+p+2)}	
		\log \left( 
		1+ \frac{q\| X\|_{_F} \| M^B \|_{_F} \sqrt{np} }{ \sqrt{(p+q) r^* }} \right) 
		+
		\log(1/\epsilon) }{\sqrt{nq}},
	\end{align*}
	where $ \Psi_{\varsigma}$ is a known constant  depending only on $\varsigma $.
\end{Proposition}

\subsection{Dealing with imcomplete responses}

As previously mentioned in the introduction, the method proposed in \cite{luo2018leveraging} has the capability to handle incomplete response data. Similarly, our proposed approach can also be easily and naturally extended to address the scenario where the response matrix $ Y $ contains missing data. The ability to handle missing data is an important consideration, as it is a common issue in many real-world datasets. This can be especially useful in cases where data collection is difficult or expensive, as it allows for the use of all available data, rather than discarding observations with missing data. 

Let $ \Omega = \{ (i,k): y_{ik} \,\, {\rm is \,\, observed } \, i \in \{ 1, \ldots, n\}, k\in \{1,\ldots, q \} \} $ be the index set of the observed entries of the binary response matrix $ Y $. Here, we have that $ |\Omega| =m < nq $. We assume that we observe a design matrix $X $ and $m $ i.i.d random pairs $ (\mathcal{O}_1, Y_1),
\ldots, (\mathcal{O}_m, Y_m) $ . The variables $ \mathcal{O}_i $ are i.i.d copies of a random variable $ \mathcal{O} $ having distribution on the set
$ \lbrace 1, \ldots, n \rbrace \times \lbrace 1, \ldots, p \rbrace $.

The risk in this case is given as
\begin{align*}
R(M) 
= 
\mathbb{E} \left[ \mathbbm{1}_{( Y_{1} \cdot (X M)_{\mathcal{O}_1} < 0)}\right],
\end{align*}
and its empirical counterpart is:
\begin{align*}
r_m(M) 
= 
\frac{1}{m}\sum_{i=1}^m
\mathbbm{1}_{(Y_{i} (XM)_{\mathcal{O}_i }<0)} .
\end{align*}
The hinge empirical risk is now as
\begin{align*}
r^h_m(M) 
= 
\frac{1}{m}\sum_{i=1}^m
(1 - Y_{i} (XM)_{\mathcal{O}_i })_{_+} \, .
\end{align*}

The subsequent theorem establishes a theoretical bound that links the integrated risk of the estimator to the minimum attainable risk achieved by the Bayes classifier, $ M^B $.

\begin{theorem}
	\label{thm_2}
	Assume that Assumption \ref{assume_margin} is satisfied and put $ r^* = {\rm rank} (M^B ) $. Then, 
	for any $\epsilon\in(0,1)$ and for $\lambda = 2 m/(3C + 2) $, $ \tau^2 =  (p+q)/(2qpm \|X\|_{F}^2 ) $, $ \upsilon \in(0,1) $,
	with probability at least $1-2\epsilon$,
	\begin{align*}
	\int R d\widehat{\rho}_\lambda 
	\leq 
	2.5\overline{R} + \Xi'_{C,\upsilon} 
	\frac{ r^* (q+p+2)	\log \left( 
		1+ \frac{q\| X\|_{_F}  \| M^B \|_{_F} \sqrt{mp} }{ \sqrt{(p+q) r^* }} \right) 
		+
		\log(1/\epsilon) }{m}
	\end{align*}
	where $ \Xi'_{C,\upsilon}$ is known constant that depends only on the  $\upsilon,C$.
\end{theorem} 

\begin{remark}
The message of Theorem \ref{thm_2} lies in its ability to give a finite sample bound and to demonstrate that our estimate remains effective in a non-asymptotic scenario when the intrinsic dimension (i.e., the rank) is relatively small in relation to $m $. To clarify this point, consider the scenario where $p $ is a function of $m $ that increases as $m $ increases. In this scenario, a traditional asymptotic approach would not provide useful information, but our bounds still provide valuable insights, as long as the rank of the parameter matrix is sufficiently small in relation to $m $.	
\end{remark}

The selection of $ \lambda $ in our results is based on optimizing an upper bound on the risk $ R $ (as presented in the proofs of the theorems, given in Appendix \ref{sc_proof}). However, it is important to keep in mind that this choice may not always be the most suitable option in practice, even though it provides a reliable estimate of the magnitude of $ \lambda $. To ensure optimal performance, it is recommended to use cross-validation to adjust the temperature parameter correctly.

\section{Numerical studies}
\label{sc_numstudy}

\subsection{Implementation and compared methods}
\subsubsection*{Implementation}
In this section, we introduce the use of the Langevin Monte Carlo (LMC) algorithm as a method for sampling from the (pseudo) posterior. 
The LMC algorithm is a gradient-based method for sampling from a  distribution. 

First, a constant step-size unadjusted LMC algorithm, as described in \cite{durmus2019high}, is proposed. The algorithm starts with an initial matrix $M_0$ and uses the recursion:
\begin{equation}
\label{langevinMC}
M_{k+1} = M_{k} - h\nabla \log \widehat{\rho}_{\lambda}(M_k) +\sqrt{2h}N_k\qquad k=0,1,\ldots
\end{equation}
where $h>0$ is the step-size and $ N_0, N_1,\ldots$ are independent random matrices with i.i.d standard Gaussian entries. When the step-size $h$ is not small enough, the sum can explode \cite{roberts2002langevin}, so a Metropolis–Hastings (MH) correction is included in the algorithm. This guarantees convergence to the desired distribution, but slows down the algorithm due to the additional acception/rejection step at each iteration.

The update rule in \eqref{langevinMC} is now considered as a proposal for a new candidate,
\begin{align}
\tilde{M}_{k+1} = M_{k} - h\nabla \log \widehat{\rho}_{\lambda} (M_k) +\sqrt{2h},N_k,\qquad
k=0,1,\ldots,
\label{mala}
\end{align}
This proposal is accepted or rejected according to the MH algorithm with probability:
$$
\min \left\lbrace 1, \frac{ \widehat{\rho}_{\lambda} (\tilde{M}_{k+1}) q(M_k | \tilde{M}_{k+1}) }{
	\widehat{\rho}_{\lambda} (M_k ) q(\tilde{M}_{k+1} | M_k ) } \right\rbrace,
$$
where
$
q(x' | x) 
\propto 
\exp \left(- \| x'-x + h\nabla \log \widehat{\rho}_{\lambda} (x) \|^2_F / (4h) \right)
$
is the transition probability density from $x$ to $x'$. This is known as the Metropolis-adjusted Langevin algorithm (MALA). It guarantees the convergence to the (pseudo) posterior and provides a way to choose the step-size $h$. Unlike random-walk MH, MALA usually proposes moves into regions of higher probability, which are more likely to be accepted. The step-size $h$ for MALA is chosen such that the acceptance rate is approximate $0.5$ following \cite{roberts1998optimal}, while the step-size for LMC in the same setting is chosen smaller than for MALA.

\subsubsection*{Compared Methods}
We will assess the effectiveness of our proposed methods by comparing them to Bayesian approaches that rely on logistic regression. More specifically, it is now assumed that $ Y|X = 1 $ with probability $ f(XM) $ and $ Y|X = -1 $ with probability $ 1-f(XM) $, where $ f(\cdot) $ is the link function, $ f(x) = e^x/(1+e^x) $. In this case the pseudo-likelihood $ \exp(-\lambda r^\ell (M)) $, with $ \lambda = nq $, is exactly equal to the likelihood of the logistic model. Here, 
$$
r^\ell (M) = \sum_{ij} {\rm logit}(YXM_{ij})/(nq) ,
\quad \text{and} \quad
{\rm logit}(u) = \log(1+e^{-u} ) 
$$ 
is the logistic loss. The prior distribution is exactly the same as in previous sections.

As studied in \cite{zhang2004statistical}, the logistic loss can serve as a convex alternative to the hinge loss for approximating the 0-1 loss. However, it is worth noting that employing the logistic loss may lead to a slower convergence rate compared to the hinge loss \cite{zhang2004statistical}.

In this study, we evaluate the performance of our proposed methods, LMC-H and MALA-H, in comparison to three other alternatives: (1) LMC-logit, (2) MALA-logit (methods based on Bayesian logistic regression) and (3) the current state-of-the-art method  mRRR. The mRRR method, which was proposed in \cite{luo2018leveraging}, is a frequentist approach and its implementation can be found in the \texttt{R} package \texttt{rrpack} \cite{rrpack}.

\subsection{Simulation studies}

We consider different scenarios of data generation to assess the performance of our method. The sample size is fixed as $ n=100 $, while we vary the dimension of the covariates and responses as $ q=8,p=12 $ and $ q=20,p=50 $. The entries of the covariate matrix $ X $ are simulated from $ \mathcal{N} (0,1) $. 
More specifically, we consider two scenarios for the true parameter matrix $ M^* $: 
\begin{itemize}
	\item first, it is a rank-2 matrix that is a product of two rank-2 matrices, i.e $ M^* = A_{p\times 2}B_{q\times 2}^\top $, where $ A $'s and $ B $'s entries are iid drawn from $ \mathcal{N} (0,1) $;
	\item second, it is an approximate rank-2 matrix. To create an approximate rank-2 matrix, we first simulate a rank-2 matrix $ M' $ as before and then add some noise as $ M^* = 2M' + N $, where entries of $ N $ are simulated from $ \mathcal{N} (0,0.1) $. 
\end{itemize}
Then we consider the following settings to obtain the responses:
\begin{itemize}
	\item Setting I: 
	$$ Y = {\rm sign}( XM^* + E ) B. $$
	\item Setting II: with $ u = XM^* + E, $ put $ p= \exp (u)/ (1 +\exp (u) ) $: 
	$$ Y_{ij} \sim  \text{Binomial}(p_{ij}).
	$$	
\end{itemize}
Here, the noise term $ (E,B) $ is varied in different scenarios which lead to different setup in each setting. It is summarized in Table \ref{tb_type_noise}. 

The LMC, MALA are run with 30000 iterations and we take the first 1000 steps as burn-in. We set the values of tuning parameters $\lambda$ and $\tau$ to 1 for all scenarios. It is important to acknowledge that a better approach would be to tune these parameters using cross validation, which could lead to improved results. The mRRR method is run with default options and that 5-fold cross validation is used to select the rank.

Each simulation setting is repeated 100 times and we report the averaged results. The results of the simulations study are detailed in Table \ref{tb_rs_1}, Table \ref{tb_rs_1_cachieu} and Table \ref{tb_rs_2_lr}, Table \ref{tb_rs_2_lr_caochieu} and the values within parentheses indicate the standard deviations associated with each misclassification error percentage.

\begin{table}
	\centering
	\caption{Summary of simulation settings.}
	\begin{tabular}{ | lllll | }
		\toprule
		Setting & Name &  $B$ & $E $ &
		\\
		\hline
		I.1 & No noise 	& $B=1$ a.s. & $E=0 $ a.s. & 
		\\
		I.2 & With noise & $B=1 $ & $E\sim \mathcal{N}(0,1) $ a.s.  & 
		\\
		I.3 & Switch 	& $B\sim 0.9 \delta_1 + 0.1 \delta_{-1}$ & $E=0$ a.s.  & 
		\\
		I.4 & Switch with noise 	& $B\sim 0.9 \delta_1 + 0.1 \delta_{-1}$ & $E\sim \mathcal{N}(0,1) $ a.s.  & 
		\\
		II.1 & logistic 	& n.a. & $E=0 $ & 
		\\
		II.2 & logistic with noise & n.a. & $E\sim \mathcal{N}(0,1) $ & 
		\\
		\bottomrule
	\end{tabular}
	\label{tb_type_noise}
\end{table}

Among the methods mentioned, the MALA algorithm with the hinge loss consistently demonstrates the lowest misclassification error rate. In fact, in some instances, it even outperforms the frequentist mRRR method by a margin of two times. This highlights the effectiveness of the proposed method, which combines the MALA algorithm and the hinge loss, for achieving accurate classification results.

\begin{table}
	\centering
	\caption{Misclassification error. $ n=100, q =8, p =12 $, rank-2.}
	\begin{tabular}{ | llllll | }
		\hline \hline
		Setting & LMC-logit (\%) & LMC-H (\%) & MALA-logit (\%) & MALA-H (\%) & mRRR (\%)
		\\
		\hline
	 \multicolumn{6}{|c|}{fully observed.}
		\\ \hline
		I.1 & 0.51 (0.43) & 0.51 (0.44) & 0.39 (0.26)& 0.15 (0.14) & 0.51 (0.41)
		\\
		I.2 & 8.89 (2.71) & 8.89 (2.68) & 7.74 (2.50)& 6.87 (2.50) & 8.93 (2.70)
		\\
		I.3 & 5.36 (1.34) & 5.32 (1.28) & 6.13 (1.09)& 5.25 (1.01) & 5.26 (1.30)
		\\
		I.4 & 11.8 (2.84) & 11.7 (2.83)& 11.3 (2.44)& 10.6 (2.42) & 11.7 (2.85)
		\\
		II.1 & 14.4 (3.25) & 14.3 (3.34) & 12.6 (2.93)& 11.6 (2.76) & 14.4 (3.30)
		\\
		II.2 & 15.4 (3.98) & 15.4 (3.99)& 13.8 (3.64)& 12.8 (3.61) &  15.4 (3.98)
		\\
		\hline
\multicolumn{6}{|c|}{10\% of data is missing.}
		\\
		\hline
		I.1 & 2.09 (1.69) & 2.14 (1.70) & 2.47 (1.75) & 3.53 (1.94) & 2.10 (1.66)
\\
I.2 & 11.1 (4.74) & 10.9 (4.72) & 11.4 (5.11) & 11.4 (5.00) & 11.0 (4.84)
\\
I.3 & 16.1 (4.03) & 16.2 (4.08) & 17.5 (4.31) & 16.8 (4.14) & 16.0 (4.06)
\\
I.4 & 20.6 (5.24) & 20.6 (5.31) & 22.4 (5.46) & 22.2 (5.41) & 20.6 (5.14)
\\
II.1 & 16.6 (5.26) & 16.6 (5.05) & 17.4 (5.26) & 18.0 (5.39) & 16.4 (5.00)
\\
II.2 & 19.2 (5.76) & 19.2 (5.37) & 19.4 (5.32) & 20.0 (5.54) & 19.3 (5.54)
	\\
		\hline
\multicolumn{6}{|c|}{30\% of data is missing.}
		\\
		\hline
		I.1 & 2.84 (1.16) & 2.78 (1.17) & 3.07 (1.25) & 3.93 (2.80) & 2.80 (1.13)
\\
I.2 & 11.3 (3.03) & 11.2 (3.03) & 11.7 (3.07) & 11.9 (3.31) & 11.2 (3.11)
\\
I.3 & 17.3 (3.06) & 17.0 (2.92) & 18.0 (2.85) & 17.7 (2.90) & 16.9 (2.92)
\\
I.4 & 23.0 (4.21) & 23.1 (3.93) & 23.4 (3.92) & 23.3 (3.98) & 22.9 (4.00)
\\
II.1 & 17.6 (4.43) & 17.5 (4.25) & 18.0 (4.33) & 18.4 (4.31) & 17.5 (4.42)
\\
II.2 & 19.7 (4.34) & 19.8 (4.47) & 20.0 (4.39) & 20.6 (4.37) & 19.9 (4.58)
		\\
		\hline
		\hline
	\end{tabular}
	\label{tb_rs_1}
	\centering
	\caption{Misclassification error. $ n=100, q =20, p =50 $, rank-2.}
	\begin{tabular}{ | llllll | }
		\hline
		\hline
		Setting & LMC-logit (\%) & LMC-H (\%) & MALA-logit (\%) & MALA-H (\%) & mRRR (\%)
		\\
		\hline
		\multicolumn{6}{|c|}{Fully observed.}
		\\ \hline
		I.1 & 0.88 (0.63)& 0.88 (0.63)& 0.70 (0.46)& 0.46 (0.36)& 0.88 (0.63)
\\
I.2 & 4.35 (1.11)& 4.34 (1.11)& 3.57 (0.91)& 2.74 (0.69)& 4.35 (1.10)
\\
I.3 & 5.51 (0.69) & 5.52 (0.68) & 7.04 (0.73)& 5.92 (0.76)& 5.51 (0.70)
\\
I.4 & 7.82 (1.19)& 7.82 (1.21)& 8.85 (1.13)& 7.46 (1.06)& 7.83 (1.19)
\\
II.1 & 7.59 (1.78)& 7.58 (1.77)& 6.20 (1.47)& 4.86 (1.17)& 7.59 (1.77)
\\
II.2 & 8.01 (1.78)& 7.99 (1.79)& 6.57 (1.46)& 5.15 (1.20)& 7.99 (1.78)
		\\
		\hline
		\multicolumn{6}{|c|}{10\% of data is missing.}
		\\
		\hline
		I.1 & 3.32 (1.34) & 3.32 (1.35) & 3.46 (1.45) & 3.65 (1.57) & 3.34 (1.36)
\\
I.2 & 7.13 (2.15) & 7.15 (2.17) & 7.53 (2.12) & 7.69 (2.32) & 7.20 (2.20)
\\
I.3 & 16.7 (2.75) & 16.7 (2.84) & 19.1 (3.00) & 19.7 (2.97) & 16.7 (2.79)
\\
I.4 & 18.3 (2.90) & 18.4 (2.89) & 21.0 (3.17) & 21.8 (3.32) & 18.2 (2.85)
\\
II.1 & 10.6 (2.45) & 10.6 (2.47) & 11.2 (2.41) & 11.2 (2.63) & 10.6 (2.45) 
\\
II.2 & 11.7 (2.57) & 11.6 (2.54) & 12.1 (2.55) & 12.2 (2.59) & 11.6 (2.50)
		\\
		\hline
		\multicolumn{6}{|c|}{30\% of data is missing.}
		\\
		\hline
		I.1 & 3.82 (0.92) & 3.81 (0.94) & 4.04 (1.01) & 3.99 (0.98) & 3.81 (0.93)
\\
I.2 & 7.98 (2.52) & 7.98 (2.53) & 8.18 (2.24) & 8.27 (1.98) & 7.97 (2.50)
\\
I.3 & 19.9 (3.94) & 19.9 (3.92) & 21.2 (2.60) & 21.4 (2.17) & 19.9 (3.93)
\\
I.4 & 22.3 (4.51) & 22.3 (4.52) & 23.1 (2.79) & 23.3 (2.51) & 22.3 (4.56)
\\
II.1 & 11.2 (2.63) & 11.2 (2.70) & 11.7 (2.40) & 11.8 (2.39) & 11.2 (2.74) 
\\
II.2 & 11.6 (2.24) & 11.6 (2.25) & 12.3 (2.36) & 12.6 (2.50) & 11.6 (2.26)
		\\
		\hline
		\hline
	\end{tabular}
	\label{tb_rs_1_cachieu}
\end{table}

\begin{table}
	\centering
	\caption{Misclassification error. $ n=100, q =8, p =12 $, approximate rank-2.}
	\begin{tabular}{ | llllll | }
		\hline\hline
		Setting & LMC-logit (\%) & LMC-H (\%) & MALA-logit (\%) & MALA-H (\%) & mRRR (\%)
		\\
		\hline
		\multicolumn{6}{|c|}{Fully observed.}
		\\ 
		\hline
		I.1 & 1.37 (2.57) & 1.36 (2.58) & 1.15 (0.85)& 0.59 (0.52) & 1.38 (2.61)
		\\
		I.2 & 7.76 (3.24) & 7.72 (3.25) & 3.90 (0.96)& 2.88 (0.87) & 7.78 (3.27)
		\\
		I.3 & 11.0 (2.25) & 11.1 (2.31) & 7.40 (0.94)& 6.14 (0.93) & 11.2 (2.39)
		\\
		I.4 & 12.5 (3.08) & 12.2 (2.91)& 8.27 (1.18)& 7.06 (1.20) & 12.4 (3.04)
		\\
		II.1 & 9.62 (2.45) & 9.59 (2.43) & 5.94 (1.51)& 4.87 (1.48) & 9.65 (2.48)
		\\
		II.2 & 10.9 (2.58) & 10.8 (2.50)& 7.04 (1.50)& 5.98 (1.49) &  10.9 (2.55)
		\\
\hline
		\multicolumn{6}{|c|}{10\% of data is missing.}
\\
\hline
		I.1 & 7.54 (3.57) & 7.49 (3.68) & 6.36 (2.67) & 5.94 (2.64) & 7.56 (3.57)
		\\
		I.2 & 12.2 (4.38) & 12.1 (4.28) & 8.66 (3.08)  & 8.24 (3.05) & 12.3 (4.49)
		\\
		I.3 & 22.3 (4.89) & 22.1 (4.93) & 19.0 (4.72) & 18.4 (4.42) & 22.4 (5.03)
		\\
		I.4 & 22.8 (5.72) & 22.8 (5.84) & 19.9 (4.34) & 19.2 (4.41) & 23.1 (5.85)
		\\
		II.1 & 14.7 (4.45) & 14.6 (4.54) & 11.3 (3.74) & 11.0 (3.91) & 14.8 (4.59)
		\\
		II.2 & 15.4 (4.78) & 15.5 (4.88) & 11.8 (3.96) & 11.9 (4.13) & 15.6 (5.02)
		\\
		\hline
		\multicolumn{6}{|c|}{30\% of data is missing.}
		\\
		\hline
		I.1 & 11.4 (5.32) & 11.4 (5.12) & 7.97 (2.03) & 7.22 (2.08) & 11.5 (5.31)
		\\
		I.2 & 14.2 (3.75) & 14.1 (3.79) & 9.62 (2.11) & 8.98 (2.16) & 14.3 (4.02)
		\\
		I.3 & 24.3 (4.09) & 24.2 (4.17) & 20.5 (2.81) & 19.9 (2.93) & 24.6 (4.36)
		\\
		I.4 & 25.9 (4.37) & 25.7 (4.15) & 21.7 (2.74) & 21.2 (2.89) & 26.1 (4.40) 
		\\
		II.1 & 15.9 (3.95) & 15.7 (3.83) & 11.9 (2.95) & 11.5 (2.80) & 16.0 (4.11)
		\\
		II.2 & 16.3 (4.22) & 16.3 (4.07) & 12.4 (2.69) & 12.0 (2.75) & 16.3 (4.13) 
		\\
		\hline
		\hline
	\end{tabular}
	\label{tb_rs_2_lr}

	\caption{Misclassification error. $ n=100, q =20, p =50 $, approximate rank-2.}
	\begin{tabular}{ | llllll | }
		\hline
		\hline
		Setting & LMC-logit (\%) & LMC-H (\%) & MALA-logit (\%) & MALA-H (\%) & mRRR (\%)
		\\
		\hline
		\multicolumn{6}{|c|}{Fully observed data. }
\\
\hline
I.1 & 13.3 (2.27)& 13.2 (2.27) & 8.61 (1.25)& 5.77 (0.68)& 13.2 (2.26)
\\
I.2 & 13.3 (2.04) & 13.3 (2.04) & 8.63 (1.07)& 5.85 (0.76)& 13.3 (2.05)
\\
I.3 & 15.7 (2.92) & 15.7 (2.93) & 12.7 (1.62)& 10.0 (1.22)& 15.7 (2.91)
\\
I.4 & 16.0 (3.12)& 16.1 (3.12)& 13.0 (1.57)& 10.3 (1.12)& 16.1 (3.13)
\\
II.1 & 13.4 (2.56)& 13.3 (2.55)& 8.73 (1.29)& 6.02 (0.82)& 13.4 (2.55)
\\
II.2 & 13.6 (2.14)& 13.6 (2.13)& 8.88 (1.24)& 6.05 (0.75)& 13.6 (2.15)
		\\
		\hline
		\multicolumn{6}{|c|}{10\% of data is missing.}
		\\
		\hline
		I.1 & 16.1 (3.04) & 16.1 (3.08) & 14.9 (3.04) & 14.2 (2.54) & 16.1 (3.05)
		\\
		I.2 & 16.5 (3.23) & 16.5 (3.17) & 15.1 (2.80) & 14.4 (2.75) & 16.5 (3.18)
		\\
		I.3 & 25.6 (4.09) & 25.6 (3.91) & 25.6 (3.25) & 25.4 (3.17) & 25.6 (3.97)
		\\
		I.4 & 25.8 (4.49) & 25.7 (451) & 25.7 (3.73) & 25.4 (3.66) & 25.7 (4.46)
		\\
		II.1 & 16.3 (3.38) & 16.4 (3.29) & 14.8 (2.73) & 14.0 (2.53) & 16.3 (3.30)
\\
II.2 & 17.0 (3.62) & 16.9 (3.57) & 15.1 (3.30) & 14.8 (3.15) & 16.9 (3.53)
		\\
		\hline
		\multicolumn{6}{|c|}{30\% of data is missing.}
		\\
		\hline
		I.1 & 17.0 (2.77) & 17.0 (2.74) & 15.8 (2.35) & 15.2 (2.10) & 17.0 (2.74)
		\\
		I.2 & 16.8 (2.57) & 16.8 (2.50) & 15.8 (2.27) & 15.2 (2.00) & 16.7 (2.58)
		\\
		I.3 & 29.7 (4.40) & 29.8 (4.43) & 28.3 (2.78) & 27.6 (2.46) & 29.8 (4.40)
		\\
		I.4 & 29.9 (4.30) & 29.9 (4.24) & 28.3 (2.71) & 27.4 (2.42) & 29.9 (4.29)
		\\
		II.1 & 17.4 (3.12) & 17.4 (3.10) & 16.4 (2.34) & 15.8 (2.13) & 17.4 (3.07)
		\\
		II.2 & 17.7 (3.20) & 17.7 (3.23) & 16.7 (2.94) & 16.1 (2.48) & 17.7 (3.25)
		\\
		\hline
		\hline
	\end{tabular}
	\label{tb_rs_2_lr_caochieu}
\end{table}

\newpage
The other methods that employ the logistic loss function, such as MALA-logit and the LMC algorithm-based approaches (LMC-H and LMC-logit), exhibit similar performance to the mRRR method. These methods are generally comparable in terms of misclassification error rates. However, it is worth noting that the MALA-logit approach, which also utilizes the MALA algorithm, shows superiority to the mRRR method in cases involving low-rank matrices and higher dimensions.

While the LMC-based methods (LMC-H and LMC-logit) perform well, they are not significantly better than the mRRR method. However, these approaches are noted to be more efficient in handling larger data sets, which can be advantageous in certain scenarios.

It's important to acknowledge that as the proportion of missing data increases to 10\% and 30\%, the misclassification error percentages also increase. This suggests a decline in the performance of these methods in the presence of missing data. Therefore, addressing and mitigating the effects of missing data is crucial for improving the performance of these methods in real-world applications.

In summary, the MALA algorithm with the hinge loss stands out as the method with consistently lower misclassification error rates. The logistic loss-based methods and the LMC-based methods demonstrate comparable performance to the mRRR method. However, it is necessary to carefully consider the impact of missing data on these methods' performance and employ appropriate strategies to handle missing data effectively.

\subsection{A real data study}

In this section, we evaluate the performance of our proposed method on a real data with multiple binary responses. To do this, we utilize the \texttt{spider} data set, which can be found in the \texttt{R} package \texttt{mvabund} \cite{wang2012mvabund}. The matrix of covariates, $ X \in \mathbb{R}^{28\times 6} $, includes information on 6 environmental features from 28 samples. The response matrix, $ Y \in \mathbb{R}^{28\times 12} $, is count data that represents the number of 12 hunting spider species from 28 observations of abundance. We convert the response data into binary format by setting $ y_{ij} = -1 $ if there is no such species surviving in a certain environment and $ y_{ij} = 1 $ otherwise. This results in a total of 154 negative ones (45.8\%) and 182 positive ones (54.2\%).

The data is divided randomly into two sets: a training set consisting of 23 samples and a test set consisting of 5 samples. We use the training data to run the methods and then evaluate their prediction accuracy based on the test data. This process is repeated 100 times, each time with a different random partition of the training and test data. The results of this procedure are illustrated in Figure \ref{fg_realdatabbrrr}. By repeating the procedure multiple times and averaging the results, we can obtain a more robust and accurate assessment of the performance of the methods. This approach allows us to account for any potential variability in the data and obtain a better understanding of the methods' performance.

\begin{figure}[!ht]
	\centering
	\includegraphics[scale=.4]{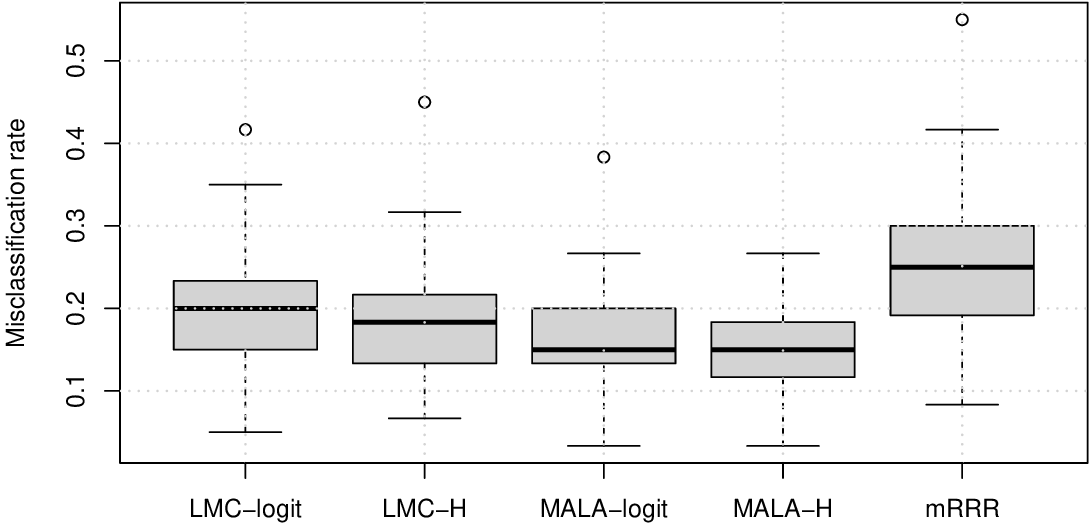}
	\caption{Result on real data.}
	\label{fg_realdatabbrrr}
\end{figure}

The results, from Figure \ref{fg_realdatabbrrr}, show that our proposed method, computed using the Metropolis-Adjusted Langevin Algorithm with the hinge loss ('MALA-H'), outperforms the frequentist mRRR method. The prediction errors for MALA-H and mRRR are 15.05\% ($ \pm $5.15\%) and 24.96\% ($ \pm $7.85\%), respectively. Other approaches, such as those using LMC, also perform well and are slightly better than the mRRR method. The approach MALA-logit, which utilizes the logit loss, is slightly behind MALA-H with a prediction error of 16.47\% ($ \pm $5.63\%). This suggests that MALA-H is the most effective method among those tested.

To further evaluate the performance of all the considered methods, we also investigate their performance when missing data is present in the response matrix of the \texttt{spider} real data set. To do this, we randomly remove 10\%, 20\%, and 30\% of the entries in the response matrix. We denote by $ \Omega $ the index set of the observed entries. We then run all the considered methods on the data set with incomplete responses and evaluate the prediction/misclassification error on the set of unobserved entries and this process is repeated 100 times. This allows us to assess how well the methods can handle missing data and make predictions even when certain information is missing. The results of this evaluation are presented in Figure \ref{fg_real_missingdata}. In general, when examining incomplete response scenarios, the outcomes are comparable to those observed in complete response cases, implying that all the methods under consideration are capable of handling missing data. Furthermore, it is worth noting that MALA-H continues to outperform other methods in dealing with incomplete response situations.

\begin{figure}[!ht]
	\centering
	\includegraphics[scale=.5]{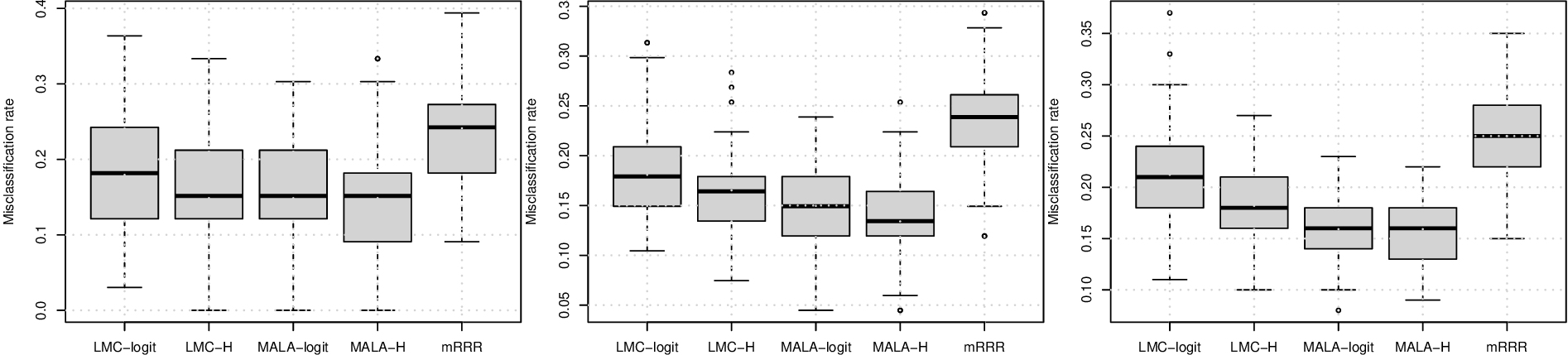}
	\caption{Results on real data with missing data in the response matrix. Left: 10\% of the entries is removed, middle: 20\% of the entries is removed, right: 30\% of the entries is removed.}
	\label{fg_real_missingdata}
\end{figure}

\section{Discussion and conclusion}
\label{sc_conclus}
In this paper, we investigated the problem of making predictions for multivariate binary responses using a set of covariates by exploring low-rank predictors through the lens of machine learning. We focused on providing the prediction error rate, which has not been addressed in previous published works. Our approach leverages methods from statistical learning theory for binary classification and does not require any assumptions about the underlying observations. Instead, we focus on a set of predictors and aim to find the one that results in the lowest prediction error. We propose using a pseudo-Bayesian method in this paper, which is also able to handle incomplete response data. 

Furthermore, we developed an efficient computational approximation method, based on a gradient-based sampling technique known as Langevin Monte Carlo. By implementing this method, we were able to overcome the computational challenges that are often associated with this type of probabilistic approach, making our proposed approach more practical and applicable in real-world settings. The numerical studies we conducted on simulated and real-world data sets have shown promising results when compared to the state-of-the-art method, further validating the effectiveness of our proposed approach. 

Although, our work offers a promising solution for the problem of relating a set covariates to multiple binary response, but there is still room for further exploration and improvement. One area of future research could include extending our method to incorporate variable selection, in order to identify the most important covariates in determining the binary responses. Additionally, future research could also focus on addressing the problem of missing data in the covariate matrix $ X $, which is a common issue in many real-world datasets. This would further improve the robustness and applicability of our proposed method. 

An additional aspect that requires further attention in practical applications is the tuning of the learning rate $ \lambda $ and the parameter $ \tau $ in the prior distribution. While we have presented certain values that yield favorable theoretical outcomes, it is important to acknowledge that these choices may not be optimal in practice, although they offer some guidance regarding the magnitude of the tuning parameters. We have also mentioned that cross-validation can be employed in practical scenarios, albeit at the cost of increased computational time. It is worth highlighting that the optimal tuning of these parameters remains a challenging problem in practical settings. In particular, tuning the learning rate $ \lambda $ remains as an open research question that has garnered significant attention within the framework of generalized Bayesian inference, see for example \cite{meunier2021meta,wu2023comparison} and references there in.

Furthermore, when dealing with practical problems involving huge datasets, it has been observed that LMC algorithms may encounter scalability issues. In order to address this challenge, variational inference (VI) has emerged as a computational optimization-based alternative to Markov chain Monte Carlo techniques. VI has gained popularity for approximating intractable posterior distributions in large-scale Bayesian models due to its comparable effectiveness and superior computational efficiency. The connection between the PAC-Bayesian approach and variational inference has been elucidated in \cite{AlquierRidgway2015}, while the development of variational inference for matrix completion has been explored in \cite{cottet20181}. These references provide a roadmap for future research, aiming to develop a more scalable computational approximation method for our proposed approach.

\subsection*{Acknowledgments}
TTM is supported by the Norwegian Research Council, grant number 309960 through the Centre for Geophysical Forecasting at NTNU. The R codes and data used in the numerical experiments are available at:  \url{https://github.com/tienmt/binary_rrr} .

\subsection*{Conflicts of interest/Competing interests}
The author declares no potential conflict of interests.

\appendix
\section{Proofs}
\label{sc_proof}

For any $\Theta\subset \mathbb{R}^{n_1 \times n_2}$,  let $\mathcal{P}(\Theta)$ denote the set of all probability distributions on $\Theta$ equipped with the Borel $\sigma$-algebra. 
For $(\mu,\nu)\in \mathcal{P}(\Theta)^2$, $\mathcal{K}(\nu,\mu) $ denotes the Kullback-Leibler divergence.
\begin{lemma}[Donsker and Varadhan's inequality, see \cite{catonibook}{Lemma 1.1.3}]
	\label{lemma_donvara}
	
	Let $\mu \in\mathcal{P}(\Theta)$. For any measurable, bounded function $h:\Theta\rightarrow\mathbb{R}$ we have:
	\begin{equation*}
	\log \int {\rm e}^{h(\theta)} \mu({\rm d}\theta) = 
	\sup_{\rho\in\mathcal{P}(\Theta)}\left[\int h(\theta) \rho({\rm d}\theta) 
	-
	\mathcal{K} (\rho , \mu)\right].
	\end{equation*}
	Moreover, the supremum w.r.t $\rho$ in the right-hand side is
	reached for the Gibbs distribution,
$
\rho (d\theta)
\propto 
\exp(h(\theta)) \pi (d\theta).
$
\end{lemma}

We will make use of the following version of the Bernstein's lemma taken from \cite[page 24]{MR2319879}.
\begin{lemma}
	\label{lemmemassart} Let $U_{1}$, \ldots, $U_{n}$ be independent real
	valued random variables. Let us assume that there are two constants
	$v$ and $w$ such that
	$
	\sum_{i=1}^{n} \mathbb{E}[U_{i}^{2}] \leq v 
	$
	and that for all integers $k\geq 3$,
	$
	\sum_{i=1}^{n} \mathbb{E}\left[(U_{i})_+^{k}\right] \leq v k!w^{k-2}/2. 
	$
	\\
	Then, for any $ \zeta \in (0,1/w)$,
	$ \mathbb{E}
	\exp\left[\zeta \sum_{i=1}^{n}\left[U_{i}-\mathbb{E}U_{i}\right]
	\right]
	\leq 
	\exp\left(\frac{v\zeta^{2}}{2(1-w\zeta )} \right) .$
\end{lemma}

Firstly, we establish a general PAC-Bayesian bound for our problem as a preliminary step. Subsequently, the specific case necessary for obtaining the result stated in Theorem \ref{th:theo} will be examined.
\begin{lemma}
	\label{le:theo2}
	Assume that Assumption \ref{assume_margin} is satisfied and that $\lambda<2nq/(C+2) $.
	Then, for $\epsilon \in (0,1) $, with probability at least $1-\epsilon $:
	\begin{align}
	\int R d\widehat{\rho}_\lambda 
	\leq 
	\overline{R} 
	+ 
	\frac{1}{1-\frac{C\lambda }{2nq(1-\lambda/nq)}} 
	\left\lbrace 
	\inf_{
		\rho \in \mathcal{P}(\mathbb{R}^{p\times q}) 
	}
	\left[ 
	r^h d\rho + \frac{\mathcal{K}(\rho,\pi)}{\lambda} 	\right] 
	-
	 \overline{r} + \frac{\log( 1/\epsilon)}{\lambda} 
	 \right\rbrace
	.
	\end{align}
\end{lemma}

\begin{proof}
	Fix any $M$ and put
	$
	U_{ij} 
	=  
	\mathbbm{1}_{Y_{ij} (XM)_{ij} \leq 0} - 
	\mathbbm{1}_{Y_{ij} (XM^B)_{ij} \leq 0} 
	.
	$
	Under Assumption \ref{assume_margin}, we have that
	$
	\sum_{ij} \mathbb{E}[U_{ij}^{2}]
	\leq
	nqC[R( M )-\overline{R}] .$
	Now, for any integer $k\geq 3$, as the 0-1 loss is bounded, we have that
	$$
	\sum_{ij} \mathbb{E}\left[(U_{ij})_+^{k}\right] \leq
	\sum_{ij} \mathbb{E}\left[ |U_{ij}|^{k-2} |U_{ij}|^{2}\right]
	\leq
	\sum_{ij} \mathbb{E}\left[ |U_{ij}|^{2}\right].
	$$ 
	Thus, we can apply Lemma~\ref{lemmemassart} with  $v := 	nqC[R( M )-\overline{R}] $, $w:=1 $ and $\zeta := \lambda/nq $. We obtain, for any $\lambda\in(0,nq) $,
	\begin{align*}
	 \mathbb{E}\exp\{\lambda (
	  [R( M )-\overline{R}]-[
	r( M )-\overline{r}] )
	\} 
	\leq 
	 \exp \left\{ \frac{C\lambda^2[R( M )-\overline{R}] }{2nq (1-\lambda/nq)} 	 
	  \right\},
	\end{align*}
	and
	\begin{align*}
	\int \mathbb{E}\exp\left\{\lambda[R( M )-\overline{R}]-\lambda[
	r( M )-\overline{r}] 
	-
	\frac{C\lambda^2[R( M )-\overline{R}] }{2nq(1-\lambda/nq)} 
	\right\} {\rm d}\pi( M ) 
	\leq 
	1.
	\end{align*}
	Them, using Fubini's theorem, we get:
	\begin{align*}
	\mathbb{E} \int \exp \left\{ 
	(\lambda- \frac{C\lambda^2 }{2nq(1-\lambda/nq)} )[R( M )-\overline{R}] 
	-\lambda[r( M )-\overline{r}]
	\right\} \pi(dM) 
	\leq 1.
	\end{align*}
	Consequently, using Lemma \ref{lemma_donvara},
	\begin{align*}
	\mathbb{E} \exp \left\lbrace \sup_{\rho} \int \left\{ 
	(\lambda- \frac{C\lambda^2 }{2nq(1-\lambda/nq)} )[R( M )-\overline{R}] 
	-\lambda[r( M )-\overline{r}]
	\right\}
	\rho (d M) - 
	\mathcal{K}(\rho,\pi) \right\rbrace 
	\leq 1.
	\end{align*}
	Using Markov's inequality,
	\begin{align*}
	\mathbb{P} \left( \sup_{\rho} \int 
	\left\lbrace 
	(\lambda-\frac{C\lambda^2 }{2nq(1-\lambda/nq)}) 
	[R( M )-\overline{R}]-\lambda [r( M )-\overline{r}]  \right\rbrace
	\rho (dM) 
	- \mathcal{K}(\rho,\pi) +\log \epsilon>0 \right)  \leq \epsilon.
	\end{align*}
	Then taking the complementary and we obtain  with probability at least $1-\epsilon$ that:
	\begin{align*}
	\forall \rho, \quad (\lambda-\frac{C\lambda^2 }{2nq(1-\lambda/nq)}) \int  
	[R( M )-\overline{R}]\rho (dM) 
	\leq 
	\lambda \int 
	[r( M )-\overline{r}]  \rho (dM) + \mathcal{K}(\rho,\pi) + \log 
	\frac{1}{\epsilon}.
	\end{align*}
	Now, note that as $ r^h \geq r $,
	\begin{align*}
	\lambda \left[ \int  r d\rho - \overline{r}_n\right] + 
	\mathcal{K}(\rho,\pi)+\log\frac{1}{\epsilon}
\leq  
	\lambda \left[ \int r^h d\rho + \frac{1}{\lambda} 
	\mathcal{K}(\rho,\pi) 
	\right] - \lambda \overline{r} +\log\frac{1}{\epsilon} 
	.
	\end{align*}
	As it stands for all $\rho$ then the right hand side can be minimized and, from Lemma \ref{lemma_donvara}, the minimizer over $\mathcal{P}(\mathbb{R}^{p\times q}) $ is $\widehat{\rho}_\lambda$.	Thus we get, when $\lambda<2nq/(C+2) $,
	\begin{align*}
	\int R d\widehat{\rho}_\lambda 
	\leq 
	\overline{R} + \frac{1}{1-\frac{C\lambda }{2nq(1-\lambda/nq)}} 
	\left\lbrace 
	\inf_{\rho \in \mathcal{P}(\mathbb{R}^{p\times q}) } 
	\left[ \int 
	r^h d\rho + \frac{1}{\lambda} \mathcal{K}(\rho,\pi) 
	\right] - \overline{r} + \frac{1}{\lambda} 
	\log\frac{1}{\epsilon} \right\rbrace
	.
	\end{align*}
\end{proof}

\subsection{Proof for Theorem \ref{th:theo}}

Finally, we consider the distributions $ \rho \in \mathcal{P}(\mathbb{R}^{p\times q}) $ that will be defined as translations of the prior $\pi$.
\begin{definition}
	\label{dfn:posterior:transla}
	For matrix $ M^B \in \mathbb{R}^{p\times q}$, we define $ \tilde{\rho}_{M^B}(M) \in \mathcal{P}(\mathbb{R}^{p\times q})$
	by 
	$$
	\tilde{\rho}_{M^B}(M) = \pi(M^B-M).
	$$
\end{definition}

\begin{proof}[\bf Proof of Theorem~\ref{th:theo}]
	We apply Lemma~\ref{le:theo2}

	\begin{align}
	\int R d\widehat{\rho}_\lambda 
	\leq 
	\overline{R} + \frac{1}{1-\frac{C\lambda }{2nq(1-\lambda/nq)}} 
	\left\lbrace \inf_{\rho \in \mathcal{P}(\mathbb{R}^{p\times q})} 
	\left[ \int r^h d\rho + \frac{1}{\lambda} \mathcal{K}(\rho,\pi) 
	\right] 
	- 
	\overline{r} + \frac{1}{\lambda} 
	\log\left(\frac{1}{\epsilon}\right) \right\rbrace
	.
	\end{align}

\noindent First, we have that,
		\begin{align*}
	\int r^h (M) \rho (dM) 
	&	=  	
	\frac{1}{nq} \int \sum_{i=1}^n \sum_{j=1}^q  ( 1 - Y_{ij} (XM)_{ij})_{_+}  \rho (dM) 
	\\
	& \leq
	\frac{1}{nq} \left[ \sum_{i=1}^n\sum_{j=1}^q  
	( 1 - Y_{ij} (XM^B)_{ij})_+
	+
	\int \sum_{i=1}^n\sum_{j=1}^q  \left| ( X ( M - M^B ) )_{ij}\right|  \rho (dM)  \right]
	\\
	& \leq
	r^h (M^B) + 	\int \| X ( M - M^B ) \|_F^2  \rho (dM) .
	\end{align*}
	And for $ \rho = \tilde{\rho}_{M^B}(M)  $, and using Lemma 1 in \cite{dalalyan2020exponential},
	\begin{align*}
	\int \| X ( M - M^B ) \|_F^2 \, \tilde{\rho}_{M^B} (dM)
	=
	\int \| X M \|_F^2  \pi (dM)
	\leq
	\|X\|_{F}^2 \int \|M\|_{F}^2 \pi({\rm d}M) 
	\leq 
	\|X\|_{F}^2 q p \tau^2. 
	\end{align*}
	From Lemma 2 in \cite{dalalyan2020exponential}, we have, with $ r^* = {\rm rank} (M^B) $, that
	$$
	\mathcal{K}( \tilde{\rho}_{M^B}(M) ,  \pi) 
	\leq 
	2 r^* (q+p+2) \log \left( 1+ \frac{\| M^B \|_F}{\tau \sqrt{2r^* }} \right) $$
	with the convention $0\log(1+0/0)=0$.
	
	As Assumption \ref{assume_trueBaye},  $r^h(M^B) \leq 2\overline{r}$, we obtain
	\begin{align*}
	\int R d\widehat{\rho}_\lambda 
	\leq 
	\overline{R} + \frac{1}{1-\frac{C\lambda }{2nq(1-\lambda/nq)}} 
	\Biggl\lbrace  \overline{r} 
 +
	 \|X\|_{F}^2 q p \tau^2
	+ 
	\frac{2 r^* (q+p+2)}{\lambda}  \log \left( 1+ \frac{\| M^B \|_F}{\tau \sqrt{2r^*}} \right) 
	+
	 \frac{1}{\lambda} \log\left(\frac{1}{\epsilon}\right) 
	\Biggl\rbrace
	.
	\end{align*}
	Then, we use Lemma~\ref{lem_theo1} to get, with probability at least
	$1-2\varepsilon$,
	\begin{align*}
	\int R d\widehat{\rho}_\lambda  
	\leq \overline{R}
	+ \frac{1}{1-\frac{C\lambda }{2nq(1-\lambda/nq)}} 
	\Biggl\lbrace \overline{R}
	+
	\frac{1}{nq \varsigma}\log \frac{1}{\epsilon} 
& + 
	\|X\|_{F}^2 q p \tau^2
	+
	\\
	&
	\frac{2 r^* (q+p+2)}{\lambda}  \log \left( 1+ \frac{\| M^B \|_F}{\tau \sqrt{2r^*}} \right) 
	+ 
	\frac{1}{\lambda} \log\left(\frac{1}{\epsilon}\right) 
	\Biggr\rbrace
	.
	\end{align*}
	Taking $\lambda = 2 nq/(3C + 2) $,
 we obtain:
	\begin{align*}
	\int R d\widehat{\rho}_\lambda  
	\leq 	2.5\overline{R}
 + 
	 1.5\|X\|_{F}^2 q p \tau^2
	+
	\frac{3(3C+2)r^* (q+p+2)}{2nq }  \log \left( 1+ \frac{\| M^B \|_F}{\tau \sqrt{2r^*}} \right) 
	+ 
	\frac{6 + 9C\varsigma + 6\varsigma }{4 nq \varsigma } \log\left(\frac{1}{\epsilon}\right)
	.
	\end{align*}
	The choice $ \tau^2 =  (p+q)/(2q^2pn \|X\|_{F}^2 ) $ leads to
	\begin{align*}
	\int R d\widehat{\rho}_\lambda  
	\leq 
	2.5 \overline{R}
	 + 1.5\frac{p+q}{2nq}
	+
	\frac{3(3C+2)r^* (q+p+2)	\log \left( 
		1+ \frac{q\| X\|_F \| M^B \|_F \sqrt{np} }{ \sqrt{(p+q)r^*}} \right) }{2nq }  
+
	\frac{6 + 9C\varsigma + 6\varsigma }{4 nq \varsigma } \log\left(\frac{1}{\epsilon}\right)
	.
	\end{align*}
	The results of Theorem \ref{th:theo} is obtained.
\end{proof}

\begin{lemma}
	\label{lem_theo1}
	For $\epsilon \in (0,1) $, with probability at least $1-\epsilon$, we have, for every $ \varsigma \in (0,1) $, that
	\begin{align*}
	\overline{r}
	\leq 
	\overline{R}+\frac{1}{nq\varsigma }\log \frac{1}{\epsilon}
	.
	\end{align*}
\end{lemma}

\begin{proof}
	Let $ \varsigma \in (0,1)$, we have that
	\begin{align*}
	\mathbb{E}\left(\exp [ \varsigma nq \overline{r}]\right) 
	&= 
	\prod_{i=1}^n \prod_{j=1}^q 
	\mathbb{E}\left(\exp \left[ \varsigma \mathbbm{1}_{(Y_{ij}  (XM^B)_{ij}<0)} \right] \right) 
	\\
	& 
	\leq 
	\prod_{i=1}^n \prod_{j=1}^q 
	\left( e^\varsigma \mathbb{E}\left[ 
	\mathbbm{1}_{(Y_{ij} (XM^B)_{ij}<0)} \right] \right) 
	\\
	& 
	\leq 
	\prod_{i=1}^n \prod_{j=1}^q 
	\left( e^\varsigma \overline{R}
	\right) 
	\leq 
	\exp \left( \varsigma nq \overline{R} \right) .
	\end{align*}
	Thus we obtain, for $\epsilon \in (0,1)$:
	\begin{align*}
	\mathbb{E}\left[ \exp \left( \varsigma nq\overline{r} -  \varsigma nq \overline{R} -\log 
	\frac{1}{\epsilon} \right)\right] \leq \epsilon.
	\end{align*}
	Now, using Markov's inequality, we get that 
	\begin{align*}
	\varsigma nq\overline{r} -  \varsigma nq \overline{R} -\log 
	\frac{1}{\epsilon}
	\leq 0 ,
	\end{align*}
	with probability at least $ 1-\epsilon $. Thus, the result of the lemma is obtained.
\end{proof}

\subsection{Proof for Proposition \ref{thm_propotion}}

\begin{proof}[\bf Proof of Proposition~\ref{thm_propotion}]
	As the 0-1 loss is bounded, we can apply the Hoeffding's Lemma . We obtain, for any $\lambda\in(0,nq) $,
	\begin{align*}
	\int \mathbb{E}\exp\left\{\lambda[R( M )-\overline{R}]-\lambda[
	r( M )-\overline{r}] 
	-
	\frac{\lambda^2 }{8nq} 
	\right\} {\rm d}\pi( M ) 
	\leq 
	1.
	\end{align*}
	Then taking the complementary and we obtain  with probability at least $1-\epsilon$ that:
	\begin{align*}
	\lambda \int  
	[R( M )-\overline{R}]\rho (dM) 
	\leq 
	\lambda \int 
	[r( M )-\overline{r}]  \rho (dM) 
	+ 
	\mathcal{K}(\rho,\pi) 
	+
	\frac{\lambda^2 }{8nq} 
	+ 
	\log \frac{1}{\epsilon}.
	\end{align*}
	Now, note that as $ r^h \geq r $,Thus we get, when $\lambda>0 $,
	\begin{align*}
	\int R d\widehat{\rho}_\lambda 
	\leq 
	\overline{R} + 
	\inf_{\rho \in \mathcal{P}(\mathbb{R}^{p\times q}) } 
	\left[ \int 
	r^h d\rho + \frac{1}{\lambda} \mathcal{K}(\rho,\pi) 
	\right] - \overline{r} 
	+
	\frac{\lambda }{8nq} 
	+ 
	\frac{1}{\lambda} 
	\log\frac{1}{\epsilon}
	.
	\end{align*}
	And for $ \rho = \tilde{\rho}_{M^B}(M) $, we proceed exactly the same as in the proof of Theorem \ref{th:theo} and obtain
	\begin{align*}
	\int R d\widehat{\rho}_\lambda 
	\leq 
	2\overline{R}
	+
	\frac{1}{nq \varsigma}\log \frac{1}{\epsilon} 
	+
	\|X\|_{F}^2 q p \tau^2
	+ 
	\frac{2 r^* (q+p+2)}{\lambda}  \log \left( 1+ \frac{\| M^B \|_F}{\tau \sqrt{2r^*}} \right) 
	+
	\frac{\lambda }{8nq} 
	+
	\frac{1}{\lambda} \log\left(\frac{1}{\epsilon}\right) 
	.
	\end{align*}
	Taking $\lambda =  2\sqrt{nq/(p+q+2)} $,
	we obtain:
	\begin{align*}
	\int R d\widehat{\rho}_\lambda  
	\leq 	
	2\overline{R}
	+ 
	\|X\|_{F}^2 q p \tau^2
	+
	r^* \sqrt{\frac{(q+p+2)}{ nq} }  \log \left( 1+ \frac{\| M^B \|_F}{\tau \sqrt{2r^*}} \right) 
		+
		\\
	\frac{1 }{4 \sqrt{nq(p+q+2)}} 
	+ 
	\frac{2 +  \varsigma\sqrt{nq(p+q+2)}  }{2 nq \varsigma } \log\left(\frac{1}{\epsilon}\right)
	.
	\end{align*}
	The choice $ \tau^2 =  (p+q)/(2q^2pn \|X\|_{F}^2 ) $ leads to
the result.
\end{proof}

\subsection{Proof for Theorem \ref{thm_2}}

We first start with preliminary lemmas.

\begin{lemma}
	\label{lem_theo2}
	For $\epsilon \in (0,1) $, with probability at least $1-\epsilon$ and for every $ \upsilon \in (0,1) $,
	\begin{align*}
	\overline{r}_m
	\leq 
	\overline{R}+\frac{1}{m \upsilon}\log \frac{1}{\epsilon}
	.
	\end{align*}
\end{lemma}
\begin{proof}
	Let $\upsilon \in (0,1)$, we have that 
	\begin{align*}
	\mathbb{E}\left(\exp [\upsilon m \overline{r}_m]\right) 
	&= 
	\prod_{i=1}^m
	\mathbb{E}\left(\exp \left[ \upsilon \mathbbm{1}_{(Y_{i}  (XM^B)_{\mathcal{O}_i }<0)} \right] \right) 
	\\
	& 
	\leq 
	\prod_{i=1}^m 
	\left( e^\upsilon \mathbb{E}\left[ 
	\mathbbm{1}_{(Y_{i}  (XM^B)_{\mathcal{O}_i }<0)} \right] \right) 
	\\
& 	\leq 
	\prod_{i=1}^m
	\left( e^\upsilon \overline{R}
	\right) 
	\leq 
	\exp \left(\upsilon m \overline{R} \right) .
	\end{align*}
	Therefore, for $\epsilon \in (0,1)$:
	\begin{align*}
	\mathbb{E}\left[ \exp \left( \upsilon m\overline{r}_m - \upsilon m \overline{R} -\log 
	\frac{1}{\epsilon} \right)\right] \leq \epsilon.
	\end{align*}
	Using Markov's inequality, we obtain the result.
\end{proof}

\begin{proof}[\bf Proof of Theorem~\ref{thm_2}]

	Assume that Assumption \ref{assume_margin} is satisfied, we proceed as in the proof for Theorem \ref{th:theo}, More specifically we carry as in the proof of Lemma \ref{le:theo2}, and obtain, for $\epsilon \in (0,1) $, with probability at least $1-\epsilon$ and for $ \lambda < 2m/(C+2) $:
	\begin{align}
	\int R d\widehat{\rho}_\lambda 
	\leq 
	\overline{R} 
	+
		\frac{1}{1-\frac{C\lambda }{2m(1-\lambda/m)}}  
	\left\lbrace \inf_{\rho \in \mathcal{P}(\mathbb{R}^{p\times q}) } 
	\left[ 
	r_m^h d\rho + \frac{\mathcal{K}(\rho,\pi)}{\lambda} 	\right] 
	-
	 \overline{r}_m + \frac{1}{\lambda} 
	\log\left(\frac{1}{\epsilon}\right) \right\rbrace
	.
	\end{align}
	Then, we have
	\begin{align*}
	\int r^h_m (M) \rho (dM) 
	&	=  	
 \int \frac{1}{m}\sum_{i=1}^m
	(1 - Y_{i} (XM)_{\mathcal{O}_i })_+   \rho (dM) 
	\\
	& \leq
	\frac{1}{m} \left[ \sum_{i=1}^m
	(1 - Y_{i} (XM^B)_{\mathcal{O}_i })_+ 
	+
	\int \sum_{i=1}^m \left| ( X ( M - M^B ) )_{\mathcal{O}_i } \right|  \rho (dM)  \right]
	\\
	& \leq
	r^h_m (M^B) + 	\int \| X ( M - M^B ) \|_F^2  \rho (dM) .
	\end{align*}
	Then, focusing on $ \rho = \tilde{\rho}_{M^B}(M) $, we use Lemma~\ref{lem_theo2} to get, with probability at least
	$1-2\varepsilon$,
	\begin{align*}
	\int R d\widehat{\rho}_\lambda  
	\leq \overline{R}
	+ \frac{1}{1-\frac{C\lambda }{2m(1-\lambda/m)}} 
	\Biggl\lbrace \overline{R}
	+
	\frac{1}{m \upsilon}\log \frac{1}{\epsilon} 
	& + 
	\|X\|_{F}^2 q p \tau^2
	+
	\\
	&
	\frac{2 r^* (q+p+2)}{\lambda}  \log \left( 1+ \frac{\| M^B \|_F}{\tau \sqrt{2r^*}} \right) 
	+ 
	\frac{1}{\lambda} \log\left(\frac{1}{\epsilon}\right) 
	\Biggr\rbrace
	.
	\end{align*}
	Taking $\lambda = 2 m/(3C + 2) $ and the choice $ \tau^2 =  (p+q)/(2qpm \|X\|_{F}^2 ) $ leads to
	the result.
\end{proof}

{\scriptsize 
%	\bibliographystyle{ieeetr}
%	\bibliography{refs_bbrr}

}

\end{document}